\PassOptionsToPackage{usenames,dvipsnames}{xcolor}
\documentclass[sigconf,nonacm,screen]{acmart}
\usepackage{relsize}
\usepackage{mathpartir}
\usepackage{listings}
\usepackage{microtype}

\newcommand{\deref}{\!\rightarrow\!}
\newcommand{\lbbar}{\{\kern-0.6ex[}
\newcommand{\rbbar}{]\kern-0.6ex\}}

\newcommand{\dyn}[2]{#1\,\text{:\textsubscript{\texttt{dyn}}}\,#2}
\newcommand{\Dyn}[2]{\dyn{#1}{\texttt{#2}}}
\newcommand{\st}[2]{#1\,\text{:\textsubscript{\texttt{st}}}\,#2}

\newcommand{\derefc}{\color{black}\deref}
\newcommand{\base}[2]{#1_{#2}}
\newcommand{\Base}[2]{#1_{\texttt{#2}}}
\newcommand{\cted}[2]{\texttt{cted}(#1,#2)}
\newcommand{\Cted}[2]{\texttt{cted}(#1,\texttt{#2})}
\newcommand{\true}{\texttt{true}}
\newcommand{\nul}{\texttt{null}}

\newcommand{\refines}[2]{\mathrel{{}^{#1}{\Rightarrow}^{#2}}}

\newcommand{\Cpp}{C\texttt{++}}

\newtheorem{theorem}{Theorem}
\newtheorem{assumption}{Assumption}
\newtheorem{lemma}{Lemma}

\lstdefinestyle{proof}{%
escapebegin=\color{NavyBlue},
language=C++,
tabsize=1,
mathescape=true,
columns=fullflexible,
morekeywords={pred,req,ens,skip,in,let,null},
basicstyle=\small
}

\makeatletter
\lst@AddToHook{OnEmptyLine}{\vspace{\dimexpr-\baselineskip+\smallskipamount}}
\makeatother

\title{Verifying \Cpp{} Dynamic Binding}

\author{Niels Mommen}
\orcid{0000-0001-9319-0947}
\affiliation{%
    \institution{imec-DistriNet Research Group, KU Leuven}
    \department{Department of Computer Science}
    \city{Leuven}
    \postcode{3000}
    \country{Belgium}
}
\email{niels.mommen@kuleuven.be}
\author{Bart Jacobs}
\orcid{0000-0002-3605-249X}
\affiliation{%
    \institution{imec-DistriNet Research Group, KU Leuven}
    \department{Department of Computer Science}
    \city{Leuven}
    \postcode{3001}
    \country{Belgium}
}
\email{bart.jacobs@kuleuven.be}

\begin{document}

\begin{abstract}
We propose an approach for modular verification of programs written in an object-oriented language where, like in \Cpp{}, the same virtual method call is bound to different methods at different points during the construction or destruction of an object. Our separation logic combines Parkinson and Bierman's abstract predicate families with essentially explicitly tracking each subobject's vtable pointer. Our logic supports polymorphic destruction. Virtual inheritance is not yet supported. We formalised our approach and implemented it in our VeriFast tool for semi-automated modular formal verification of \Cpp{} programs.
\end{abstract}

\maketitle

\section{Introduction}

Despite the rise of safer alternatives like Rust, \Cpp{} is still an extremely widely-used language, often for code that is safety- or security-critical \cite{dotnet_CLR,fuchsia,openjdk}. Modular formal verification can be a powerful tool for gaining assurance that programs satisfy critical safety or security requirements; however, so far no modular formal verification approaches have been proposed for \Cpp{} programs. There has been much work on modular verification of C programs, and on modular verification of object-oriented languages, including languages with multiple inheritance. However, these are not directly applicable to \Cpp{}, in large part due to its peculiar semantics of dynamic binding during object construction and destruction. In this paper, we propose what we believe to be the first Hoare logic \cite{DBLP:journals/cacm/Hoare69} for an object-oriented language that reflects \Cpp{}'s semantics of dynamic binding in the presence of constructors and destructors. Our separation logic \cite{ReynoldsJ.C.2002Slal} combines Parkinson and Bierman's abstract predicate families \cite{ParkinsonMatthew2005Slaa,ParkinsonMatthewJ.2008Slaa} with essentially explicitly tracking each subobject's vtable pointer. Our logic also supports polymorphic destruction (applying the $\mathbf{delete}$ operator to an expression whose static type is a supertype of its dynamic type). Virtual inheritance, however, is not yet supported.

The remainder of this paper is structured as follows. In \S\ref{sec:lang} we introduce the syntax and operational semantics of the minimal \Cpp-like language that we will use to present our approach. In \S\ref{sec:logic} we introduce our separation logic. In \S\ref{sec:proof_outline} we illustrate an example annotated with a proof outline of our program logic. We end with a discussion of related work (\S\ref{sec:related-work}) and a conclusion (\S\ref{sec:conclusion}).

\section{A Minimal \Cpp-like Language}\label{sec:lang}

\begin{figure}
\centering
\begin{equation*}
\begin{array}{r l}
v \Coloneqq & \texttt{null} \mid o
\\
e \Coloneqq & v \mid x \mid e \deref f \mid \texttt{new}\ C(\overline{e}) \mid ({C*})\ e
\\
c \Coloneqq & \texttt{let}\ x \coloneqq e\ \texttt{in}\ c \mid \texttt{delete}\ e \mid e \deref f \coloneqq e \\
&\mid e \deref C{::} m(\overline{e}) \mid e \deref m(\overline{e}) \mid c; c \mid \texttt{skip}
\\
field \Coloneqq & f \coloneqq \texttt{null};
\\
pred \Coloneqq & \texttt{pred}\ p (\overline{x}) = P;
\\
meth \Coloneqq & \texttt{virtual}\ m(\overline{x})\ \texttt{req}\ P\ \texttt{ens}\ Q\ {\{}c{\}}
\\
ctor \Coloneqq & C(\overline{x})\ \texttt{req}\ P\ \texttt{ens}\ Q : \overline{C(\overline{e})}\ {\{}c{\}}
\\
dtor \Coloneqq & \texttt{virtual}\ \texttt{\textasciitilde}C()\ \texttt{req}\ P\ \texttt{ens}\ Q\ {\{}c{\}}
\\
class \Coloneqq & \texttt{class}\ C : \overline{C}\ {\{} \overline{field}\ \overline{pred}\ ctor\ dtor\ \overline{meth} {\}};
\\
prog \Coloneqq & \overline{class}\ c
\end{array}
\end{equation*}
\caption{Syntax of the minimal language}
\label{fig:syntax_of_formal_language}
\end{figure}

The syntax of our minimal object-oriented programming language is shown in Fig.~\ref{fig:syntax_of_formal_language}. We assume infinite disjoint sets $\mathcal{C}$ of class names, $\mathcal{M}$ of method names, $\mathcal{F}$ of field names, and $\mathcal{X}$ of variable names, ranged over by symbols $C$, $m$, $f$, and $x$, respectively. We assume $\texttt{this} \in \mathcal{X}$. For now, we also assume a set $\mathcal{A}$ of assertions, ranged over by $P$ and $Q$. We will define the syntax of assertions in \S\ref{sec:logic}.

A program consists of a sequence of class definitions, followed by a command that gets executed when the program starts.
For the remainder of the formal treatment, we fix a program $prog$. Whenever we use a class $class$ as a proposition, we mean $class \in prog$.

For all $C$, we define the set $bases(C)$ as the set of all direct base classes of $C$: 
\begin{equation*}
\texttt{class}\ C : \overline{C}\ {\{} \cdots {\}} \Rightarrow bases(C) = \{\overline{C}\}
\end{equation*}

An \emph{object pointer} $o \in \mathcal{O}$ is either an \emph{allocation pointer} of the form $(id : C{*})$ where $id \in \mathbb{N}$ is an \emph{allocation identifier}, or a \emph{subobject pointer} of the form $o_C$ where $o$ is an object pointer:
$$ o \in \mathcal{O} \Coloneqq (id : C{*})\ |\ o_C $$

We use notation $\st{o}{C}$ to denote the object pointed to by $o$ has static type $C$:
$$\st{o}{C} \Leftrightarrow (\exists id.\;o = (id: C{*})) \vee (\exists o'.\ o = o'_C)$$

Notice that for simplicity, the values of our language are only the object pointers and the \nul{} value. Furthermore, fields and other variables are untyped and hold scalar values only. That is, objects never appear on the stack or as (non-base) subobjects of other objects.

We define a heap, ranged over by $h$, as a finite set of resources. Resources, ranged over by $\alpha$, are defined as follows:
\begin{equation*}
\alpha \Coloneqq \texttt{alloc}(id) \mid \texttt{cted}(o) \mid o \deref f \mapsto v \mid \dyn{o}{C}
\end{equation*}
where $\texttt{alloc}(id)$ means that an object with allocation identifier $id$ has been allocated, $\texttt{cted}(o)$ means that the object pointed to by $o$ (always an allocation pointer) has been fully constructed and is not yet being destructed. Resource $o \deref f \mapsto v$ means that field $f$ of the object pointed to by $o$ has value $v$, and $\dyn{o}{C}$ means that the dynamic type of the object pointed to by $o$ (always a \emph{leaf object}, whose class has no bases)\footnote{This corresponds to the fact that in \Cpp{}, objects that have polymorphic base subobjects can reuse the (first) polymorphic base subobject's vtable pointer. Note: in this paper, for simplicity we do not consider non-polymorphic classes, i.e.~classes that do not declare or inherit any virtual members.} is $C$.

We define $dtype(o, C)$ as the set of all $\dyn{}{}$ resources of its leaf base objects, or its own $\dyn{}{}$ resource when it does not have any base objects, given that $\st{o}{C'}$:
\begin{equation*}
dtype(o, C) \overset{\operatorname{def}}{=}\left\{
\begin{array}{l l}
\{ \dyn{o}{C} \} & bases(C') = \emptyset \\
\underset{1 \leq i \leq n}{\bigcup} dtype(o_{C_i}, C) & bases(C') = C_{1} \dots C_{n}
\end{array}
\right.
\end{equation*}

We say an object pointed to by $o$ has dynamic type $C$ in a heap $h$ if and only if $dtype(o, C) \subseteq h$. Notice that a non-leaf object has dynamic type $C$ if and only if all of its bases have dynamic type $C$. As we will see, dynamically dispatched calls on an object $o$ are dispatched to the dynamic type of $o$. If an object $o$ has no dynamic type in our language, dynamically dispatched calls get stuck. As we will also see, an object $o$ has no dynamic type while its bases are being constructed or destructed, nor while unrelated (i.e.~neither enclosed nor enclosing) subobjects of the allocation are being constructed or destructed. It has a dynamic type only while its own constructor's or destructor's body, or the body of an enclosing object's constructor or destructor is executing, and between the point where its enclosing allocation is fully constructed and the point where it starts being destructed.

We use $o \downarrow C$ (\emph{$o$ downcast to $C$}) to denote the pointer to the enclosing object of class $C$ of the object pointed to by $o$:
\begin{mathpar}
\inferrule
{
\st{o}{C}
}
{
o \downarrow C = o
}
\and
\inferrule
{
o \downarrow C = o'
}
{
o_{C'} \downarrow C = o'
}
\end{mathpar}
\begin{figure*}
\begin{mathpar}
\inferrule[OUpcast]
{
h, e \Downarrow  h', o \\
\st{o}{C} \\
C' \in bases(C) \\
}
{
h, ({C'*})\ e \Downarrow  h', o_{C'}
}
\and
\inferrule[ONew]
{
o = (id: C{*}) \\
id = \text{min}\{id\ |\ \texttt{alloc}(id) \notin h\}\\\\
h \uplus \lbbar \texttt{alloc}(id) \rbbar, o \deref C(\overline{e}) \Downarrow  h'
}
{
h, \texttt{new}\ C(\overline{e}) \Downarrow  h' \uplus \lbbar \texttt{cted}(o) \rbbar, o
}
\and
\inferrule[ODelete]
{
o' = o \downarrow C \\
h, e \Downarrow  h'  \uplus \lbbar \texttt{cted}(o') \rbbar, o \\\\
h', o' \deref \texttt{\textasciitilde}C() \Downarrow h''
}
{
h, \texttt{delete}(e) \Downarrow h''
}
\and
\inferrule[OStaticDispatch]
{
\texttt{class}\ C \cdots {\{} \cdots \texttt{virtual}\ m(\overline{x}) {\{} c {\}} \cdots {\}}\\
\st{o}{C} \\\\
h, e \Downarrow  h', o \\
 h', \overline{e} \Downarrow h'', \overline{v} \\
h'', c[o/\texttt{this},\overline{v}/\overline{x}] \Downarrow h'''
}
{
h, e \deref C\text{::}m(\overline{e}) \Downarrow h'''
}
\and
\inferrule[ODynamicDispatch]
{
\texttt{class}\ C \cdots {\{} \cdots \texttt{virtual}\ m(\overline{x}) {\{} c {\}} \cdots {\}} \\
h, e \Downarrow  h', o \\
o' = o \downarrow C \\\\
 h', \overline{e} \Downarrow h'', \overline{v} \\
dtype(o, C) \subseteq h'' \\
h'', c[o'/\texttt{this},\overline{v}/\overline{x}] \Downarrow h'''
}
{
h, e \deref m(\overline{e}) \Downarrow h'''
}
\and
\inferrule[OConstruct]
{
\texttt{class}\ C : C_1 \ldots C_n {\{} \overline{f \coloneqq \texttt{null};} \cdots C(\overline{x}) : C_1(\overline{e_1}) \ldots C_n(\overline{e_n}) {\{} c {\}} \cdots {\}} \\\\
h, \overline{e} \Downarrow h_0, \overline{v} \\\\
h_0, o_{C_1} \deref C_1(\overline{e_1}[o/\texttt{this},\overline{v}/\overline{x}]) \Downarrow h_1 \uplus dtype(o_{C_1}, C_1) \\\\
\ldots \\\\
h_{n-1}, o_{C_n} \deref C_n(\overline{e_n}[o/\texttt{this},\overline{v}/\overline{x}]) \Downarrow h_n \uplus dtype(o_{C_n}, C_n) \\\\
h_n \uplus \lbbar \overline{o \deref f \mapsto \texttt{null}} \rbbar \uplus dtype(o, C), c[o/\texttt{this},\overline{v}/\overline{x}] \Downarrow  h'
}
{
h, o \deref C(\overline{e}) \Downarrow  h'
}
\and
\inferrule[ODestruct]
{
\texttt{class}\ C : C_1 \ldots C_n {\{} \overline{f \coloneqq \texttt{null};} \cdots \texttt{virtual}\ \texttt{\textasciitilde}C() {\{} c {\}} \cdots {\}} \\\\
h, c[o/\texttt{this}] \Downarrow h_n \uplus dtype(o, C) \uplus \lbbar \overline{o \deref f \mapsto v} \rbbar \\\\
h_n \uplus dtype(o_{C_n}, C_n), o_{C_n} \deref \texttt{\textasciitilde}C_n() \Downarrow h_{n-1} \\\\
\ldots \\\\
h_1 \uplus dtype(o_{C_1}, C_1), o_{C_1} \deref \texttt{\textasciitilde}C_1() \Downarrow h_0
}
{
h, o \deref \texttt{\textasciitilde}C() \Downarrow h_0
}
\end{mathpar}
\caption{Operational semantics of the minimal language related to allocation and deallocation, construction and destruction, and method dispatching.}\label{fig:opsem}
\end{figure*}

We use $h, e \Downarrow h', v$ to denote that when evaluated in heap $h$, expression $e$ evaluates to value $v$ and post-heap $h'$. Similarly, we use $h, c \Downarrow h'$ and $h, o\deref C(\overline{e}) \Downarrow h'$ and $h, o\deref\texttt{\textasciitilde}C() \Downarrow h'$ to denote that command $c$, constructor call $o\deref C(\overline{e})$, and destructor call $o\deref\texttt{\textasciitilde}C()$, when executed in heap $h$, terminate with post-heap $h'$, respectively. These judgments are defined by mutual induction; we show selected rules in Fig.~\ref{fig:opsem}. (The complete set of rules can be found in the appendix.)

Notice, first of all, that a statically dispatched call $e\deref C{::}m(\overline{e})$ gets stuck if class $C$ does not declare a method $m$, even if some base does declare such a method: in our minimal language, classes do not inherit methods from their bases. The same holds for dynamically dispatched calls.\footnote{Of course, a program that does rely on method inheritance can be trivially translated into  our minimal language by inserting overrides that simply delegate to the appropriate base. Importantly, however, those overrides will have to be verified as part of the correctness proof (see \S\ref{sec:logic}); their correctness does not hold automatically.}

Evaluation of $\texttt{new}\ C(\overline{e})$ picks an unused allocation identifier $id$ and produces (i.e.~adds to the heap) $\texttt{alloc}(id)$ to mark it as used, then executes the constructor call, and finally produces $\texttt{cted}(o)$ to mark $o$ as fully constructed.

Executing a constructor call $o\deref C(\overline{e})$ is somewhat involved. If $C$ has no bases, the argument expressions are evaluated, the fields are produced, $\dyn{o}{C}$ is produced, and the constructor body is executed. Considered together with \textsc{ODynamicDispatch}, this means that dynamically dispatched calls on $\texttt{this}$ in the constructor body are dispatched to class $C$ itself, even if $C$ is not the most derived class of the allocation.

Now consider the case where $C$ does have bases. Executing constructor call $o\deref C(\overline{e})$ evaluates the argument expressions and then executes each base class' constructor on the corresponding base subobject. After executing the constructor for base $C_i$, $dtype(o_{C_i}, C_i)$ is consumed (i.e.~removed from the heap); after all base subobjects have been initialized, $dtype(o, C)$ is produced. This means that, during execution of the body of the constructor of class $C$, dynamically dispatched calls on $o$ or on any base subobject of $o$ are dispatched to class $C$. After an allocation of class $C$ is fully constructed, and until it starts being destructed, its dynamic type (and that of all of its subobjects) is $C$.

Execution of a destructor call $o \deref \texttt{\textasciitilde}C()$ performs the exact reverse process: it executes the destructor body, consumes $dtype(o, C)$ and the fields, and destructs the base subobjects. Before destructing the subobject for base $C_i$, $dtype(O_{C_i}, C_i)$ is produced, so that during execution of the body of the destructor of an object $o$ of class $C$, dynamically dispatched calls on $o$ are dispatched to class $C$. After destruction of an allocation completes, only the $\texttt{alloc}$ resource remains, to ensure that no future allocation is assigned the same identifier.\footnote{This reflects the fact that pointers in \Cpp{} become invalid permanently after the allocation they point to is deallocated, even if some future allocation happens to reuse the same address.}

Deleting an object gets stuck unless its enclosing allocation is fully constructed and is not yet being destructed, as indicated by the presence of the \texttt{cted} resource. Since this resource always holds an allocation pointer, it is always the entire allocation that is destroyed, even if the argument to \texttt{delete} is a pointer to a subobject.

We use judgments $h, e\,\texttt{div}$ and $h, c\,\texttt{div}$ and $h, o\deref C(\overline{e})\,\texttt{div}$ and $h, o\deref\texttt{\textasciitilde}C()\,\texttt{div}$ to denote that an expression, command, constructor call, or destructor call diverges (i.e.~runs forever without terminating or getting stuck), respectively. These judgments' definitions can be derived mechanically \cite{dagnino} from the definitions of the termination judgments and are therefore elided.

\section{A Program Logic for \Cpp{} Dynamic Binding}\label{sec:logic}
A class definition in our language includes a list of abstract predicates. A predicate declaration in a class defines its entry for the corresponding predicate family, i.e., a class defines its own definition for the abstract predicate, which can be overridden by derived classes. As we will see, predicate assertions involve a \emph{class index} to refer to the definition of the predicate declared in that class.

We use a context $\Gamma$, which is a sequence of class definitions.

\subsection{Assertions}
Predicate definitions, method specifications, constructor specifications, and destructor specifications consist of assertions, ranged over by $P$ and $Q$:
$$\begin{array}{r l}
P, Q \Coloneqq & \texttt{true} \mid \texttt{false} \mid P \wedge Q \mid P \vee Q \mid P \ast Q \mid \exists x.\ P \\
& \mid \varepsilon \deref f \mapsto \varepsilon \mid \varepsilon \deref p_\varepsilon(\overline{\varepsilon}) \mid \texttt{cted}(\varepsilon, \varepsilon) \mid \dyn{\varepsilon}{\varepsilon} \\
\nu \Coloneqq & v \mid C\\
\varepsilon \Coloneqq & x \mid \nu
\end{array}$$
where $P \ast Q$ is the separating conjunction of assertions $P$ and $Q$, which informally means that assertion $P$ and $Q$ must be satisfied in disjoint portions of the heap. Assertion $\varepsilon \deref p_{\varepsilon'}(\overline{\varepsilon''})$ is a predicate assertion $p$ with class index $\varepsilon'$ on the target object pointed to by $\varepsilon$.

We show the semantics of the most interesting assertions:
\begin{equation*}
\begin{array}{l r l}
I,h \vDash o \deref p_C(\overline{\nu}) & \Leftrightarrow & \exists o'.\;o \downarrow C = o' \land (h, o', p, C, \overline{\nu}) \in I \\
I,h \vDash \texttt{cted}(o, C) & \Leftrightarrow & \exists o'.\;o \downarrow C = o' \land \texttt{cted}(o') \in h \\
I,h \vDash \dyn{o}{C} & \Leftrightarrow & dtype(o, C) \subseteq h \\
I,h \vDash o \deref f \mapsto v & \Leftrightarrow & o \deref f \mapsto v \in h
\end{array}
\end{equation*}
where $I,h \vDash P$ means that assertion $P$ is satisfied, given heap $h$ and interpretation of predicates $I$. An interpretation of predicates is the least fixpoint of the program's predicate definitions considered together.

We define the assertion weakening relation $\Gamma \vdash P \Rightarrow_a Q$ by induction, where every judgment $P \Rightarrow_a Q$ should be read as $\Gamma \vdash P \Rightarrow_a Q$: 
\begin{mathparpagebreakable}
\inferrule[ADyntype]
{
\st{o}{C}\\
bases(C) = C_1 \ldots C_n\\
n > 0
}
{
\dyn{o}{C'} \Leftrightarrow_a \dyn{o_{C_1}}{C'} \ast \ldots \ast \dyn{o_{C_n}}{C'}
}
\and
\inferrule[AFrame]
{
P \Rightarrow_a P'
}
{
P \ast Q \Rightarrow_a P' \ast Q
}
\and
\inferrule[ATrans]
{
P \Rightarrow_a P'\\
P' \Rightarrow_a P''
}
{
P \Rightarrow_a P''
}
\and
\inferrule[AMoveCted]
{
\st{o}{C}\\\\
C' \in bases(C)\\
C' \neq C''
}
{
\texttt{cted}(o, C'') \Leftrightarrow_a \texttt{cted}(o_{C'}, C'')
}
\and
\inferrule[AImply]
{
\forall I,h.\ I,h \vDash P \Rightarrow I,h \vDash P'
}
{
P \Rightarrow_a P'
}
\and
\inferrule[AMovePred]
{
\st{o}{C}\\\\
C' \in bases(C)\\
C' \neq C''
}
{
o \deref p_{C''}(\overline{\nu}) \Leftrightarrow_a o_{C'} \deref p_{C''}(\overline{\nu})
}
\and
\inferrule[APredDef]
{
\st{o}{C}\\
\texttt{class}\ C \cdots {\{} \cdots \texttt{pred}\ p(\overline{x}) = P \cdots {\}} \in \Gamma
}
{
o \deref p_C(\overline{\nu}) \Leftrightarrow_a P[o/\texttt{this},\overline{\nu}/\overline{x}]
}
\end{mathparpagebreakable}

Weakening rule \textsc{APredDef} allows to switch between a predicate assertion and the definition of the predicate corresponding to the class index. The class index must be a class name declared in the program.

\textsc{AMovePred} and \textsc{AMoveCted} allow to \emph{transfer} predicate and \texttt{cted} assertions between base and derived objects. It is not possible to transfer such an assertion to an object whose dynamic type is a subtype of the predicate index and allocation class, respectively.

Weakening rule \textsc{ADyntype} states that the dynamic type assertion of a non-leaf object can be exchanged for all dynamic type assertions of its direct base objects. This means that the dynamic type of a base object can be retrieved if the dynamic type of its direct derived object is known. The other way around, it is possible to derive the dynamic type of a derived object if the dynamic type of all its direct base classes is known.

\subsection{Expression and command verification}
The verification rules for the most interesting expressions and commands are listed in Fig.~\ref{fig:verrules}, together with the verification rules for constructor and destructor invocations. These rules are related to object allocation and deallocation, and static and dynamic dispatching. (The complete set of verification rules can be found in the appendix). 

In method and destructor specifications, we use special variable $\theta$ to refer to the class of the target object of the call. This variable is assumed to be equal to the containing class during verification of the method or destructor. This is sound, because we require that a class overrides all methods of all its direct base classes, as we will later see. Hence when a call is dynamically dispatched, it will always be bound to the method declared in the class corresponding with the dynamic type of the target object.

Variable $\theta$ is substituted with the dynamic type of the target object and the static type of the target object during verification of dynamically dispatched calls and statically dispatched calls, respectively. This mechanism allows to use the specification for the method or destructor in the class corresponding to the static type of the method or destructor target.

\begin{figure*}
\begin{mathpar}
\inferrule[HStaticDispatch]
{
\texttt{class}\ C \cdots {\{} \cdots \texttt{virtual}\ m(\overline{x})\ \texttt{req}\ P\ \texttt{ens}\ Q \cdots {\}} \in \Gamma \\
\st{o}{C}
}
{
\{ P[o/\texttt{this},C/\theta,\overline{v}/\overline{x}]\}\ o \deref C\text{::}m(\overline{v})\ \{ Q[o/\texttt{this}, C/\theta, \overline{v}/\overline{x}] \}
}
\and
\inferrule[HNew]
{
\texttt{class}\ C \cdots {\{} \cdots C(\overline{x})\ \texttt{req}\ P\ \texttt{ens}\ Q \cdots {\}} \in \Gamma
}
{
\{ P[\overline{v}/\overline{x}] \}\ \texttt{new}\ C(\overline{v})\ \{ Q[\overline{v}/\overline{x}, \texttt{result}/\texttt{this}] \ast \texttt{cted}(\texttt{result}, C) \}
}
\and
\inferrule[HDynamicDispatch]
{
\texttt{class}\ C \cdots {\{} \cdots \texttt{virtual}\ m(\overline{x})\ \texttt{req}\ P\ \texttt{ens}\ Q \cdots {\}} \in \Gamma \\
\st{o}{C}
}
{
\{ \dyn{o}{C'} \wedge P[o/\texttt{this}, C'/\theta, \overline{v}/\overline{x}] \}\ o \deref m(\overline{v})\ \{ Q[o/\texttt{this}, C'/\theta, \overline{v}/\overline{x}] \}
}
\and
\inferrule[HDestruct]
{
\texttt{class}\ C \cdots {\{} \cdots \texttt{virtual \textasciitilde} C()\ \texttt{req}\ P\ \texttt{ens}\ Q \cdots {\}} \in \Gamma
}
{
\{ P[o/\texttt{this}, C/\theta] \}\ o \deref \texttt{\textasciitilde}C()\ \{ Q \}
}
\and
\inferrule[HDelete]
{
\st{o}{C} \\\\
\texttt{class}\ C \cdots {\{} \cdots \texttt{virtual \textasciitilde} C()\ \texttt{req}\ P\ \texttt{ens}\ Q \cdots {\}} \in \Gamma
}
{
\{ \texttt{cted}(o, C') \ast P[o/\texttt{this}, C'/\theta] \}\ \texttt{delete}(o)\ \{ Q \}
}
\and
\inferrule[HUpcast]
{
\st{o}{C'}\\
C \in bases(C')
}
{
\{ P[o_C/\texttt{result}] \}\ (C{*})\ o\ \{ P \}
}
\and
\inferrule[HConstruct]
{
\texttt{class}\ C \cdots {\{} \cdots C(\overline{x})\ \texttt{req}\ P\ \texttt{ens}\ Q \cdots {\}} \in \Gamma
}
{
\{ P[\overline{v}/\overline{x}]) \}\ o\deref C(\overline{v})\ \{ Q[\overline{v}/\overline{x}, o/\texttt{this}] \}
}
\end{mathpar}
\caption{Verification rules related to allocation and deallocation, construction and destruction, and method dispatching. Read judgment $\{ P \}\ c\ \{ Q \}$ as $\Gamma \vdash \{ P \}\ c\ \{ Q \}$.}\label{fig:verrules}
\end{figure*}

\subsection{Constructor verification}
The verification rule for constructors follows \textsc{OConstruct} from our operational semantics: the direct base constructor invocations are verified in order of inheritance, prior to initializing the fields of the object and verifying the command in the constructor's body. Virtual calls are always dispatched to the (sub)object under construction.
\begin{mathpar}
\inferrule
{
\forall \st{o}{C}, \overline{v}.\\\\
P[\overline{v}/\overline{x}] = P_0 \\
\{ P_0 \}\ o_{C_1} \deref C_1(\overline{e_1}[o/\texttt{this},\overline{v}/\overline{x}])\ \{ P_1 \ast \dyn{o_{C_1}}{C_1} \} \\\\
\ldots \\\\
\{ P_{n-1} \}\ o_{C_n} \deref C_n(\overline{e_n}[o/\texttt{this},\overline{v}/\overline{x}])\ \{ P_n \ast \dyn{o_{C_n}}{C_n} \} \\\\
\{ P_n \ast \overline{o \deref f \mapsto \texttt{null}} \ast \dyn{o}{C} \}\ c[o/\texttt{this},\overline{v}/\overline{x}]\ \{ Q[o/\texttt{this}, \overline{v}/\overline{x}] \}
}
{
\Gamma \vdash C(\overline{x})\ \texttt{req}\ P\ \texttt{ens}\ Q : C_1(\overline{e_1}) \ldots C_n(\overline{e_n})\ {\{} c {\}}\ \text{correct in $C$}
}
\end{mathpar}

\subsection{Behavioral subtyping}
We follow Parkinson and Bierman's approach \cite{ParkinsonMatthewJ.2008Slaa} to check whether specifications of overriding methods satisfy behavioral subtyping. A specification $\{ P_D \}\texttt{\_}\{ Q_D \}$ of an overriding method in derived class $D$ \emph{implies} a specification $\{ P_B \}\texttt{\_}\{ Q_B \}$ of a method in base class $B$, if for all commands $c$, values $\overline{v}$ and object pointers $\st{o}{B}$ with a well-defined downcast $o' = o \downarrow D$ that satisfy $\{ P_D[S_D] \}\ c\ \{ Q_D[S_D] \}$, it holds that $\{ P_B[S_B] \}\ c\ \{ Q_B[S_B] \}$ is also satisfied, with $S_B = o/\texttt{this},D/\theta,\overline{v}/\overline{x}$ and $S_D = o'/\texttt{this},D/\theta,\overline{v}/\overline{x}$. This holds when a proof tree exists using the structural rules of Hoare and Separation logic, with leaves $\Gamma \vdash \{ P_D[S_D] \}\texttt{\_}\{ Q_D[S_D] \}$ and root $\Gamma \vdash \{ P_B[S_B] \}\texttt{\_}\{ Q_B[S_B] \}$:
\begin{mathpar}
\inferrule*
{
    \Gamma \vdash \{ P_D[S_D] \}\texttt{\_}\{ Q_D[S_D] \}
}
{
    \inferrule*
    {
        \vdots
    }
    {
        \Gamma \vdash \{ P_B[S_B] \}\texttt{\_}\{ Q_B[S_B] \}
    }
}
\end{mathpar}
We use notation $\Gamma \vdash \{ P_D \}\texttt{\_}\{ Q_D \} \refines{D}{B} \{ P_B \}\texttt{\_}\{ Q_B \}$ to denote that such a proof exists.

\subsection{Method verification}
The verification rule for correctly overriding a method checks that (1) the specification for method $m$ in derived class $C$ satisfies behavioral subtyping for base class $C'$ which also declares $m$, and (2) recursively checks this condition for all direct base classes of $C'$. We use $methods(C)$ to denote all methods declared in class $C$.
\begin{mathpar}
\inferrule
{
\texttt{class}\ C \cdots {\{} \cdots \texttt{virtual}\ m(\overline{x})\ \texttt{req}\ P\ \texttt{ens}\ Q \cdots {\}} \in \Gamma\\
\texttt{class}\ C' \cdots {\{} \cdots \texttt{virtual}\ m(\overline{x})\ \texttt{req}\ P'\ \texttt{ens}\ Q' \cdots {\}} \in \Gamma\\
\Gamma \vdash \{ P \}\texttt{\_}\{ Q \} \refines{C}{C'} \{ P' \}\texttt{\_}\{ Q' \}\\
{
\begin{array}{r}
\forall C'' \in bases(C').\,m \in methods(C'') \Rightarrow\quad\\
\Gamma \vdash \text{override of $m$ in $C''$ correct in C}
\end{array}
}
}
{
\Gamma \vdash \text{override of $m$ in $C'$ correct in $C$}
}
\end{mathpar}
Method $m$ in class $C$ is correct if (1) the override check for all base classes of $C$ that declare $m$ succeeds and (2) the method body satisfies its specification given that the target class type is $C$.
\begin{mathpar}
\inferrule
{
{
\begin{array}{l}
\forall C' \in bases(C).\,m \in methods(C') \Rightarrow\\
\quad\Gamma \vdash \text{override of $m$ in $C'$ correct in C}
\end{array}
}\\\\
\forall \st{o}{C}, \overline{v}.\\\\ \{ P[o/\texttt{this}, C/\theta, \overline{v}/\overline{x}] \}\ c[o/\texttt{this}, \overline{v}/\overline{x}]\ \{ Q[o/\texttt{this}, C/\theta, \overline{v}/\overline{x}] \}
}
{
\Gamma \vdash m(\overline{x})\ \texttt{req}\ P\ \texttt{ens}\ Q\ {\{} c {\}}\ \text{correct in $C$}
}
\end{mathpar}

\subsection{Destructor verification}
The verification rule for correctly overriding a destructor is similar to the verification rule for correctly overriding a method. The difference is that it recursively checks the rule for \emph{all} bases because every class must declare a destructor in our language.

The verification rule for destructors again resembles the operational semantics and follows the reverse process of its corresponding constructor. The command of the body is first verified, followed by the removal of the object's fields and verification of the direct base destructor invocations in reverse order of inheritance. Virtual member invocations are dispatched to the (sub)object under destruction.
\begin{mathpar}
\inferrule
{
\forall C' \in bases(C).\;
\Gamma \vdash \text{override of destructor in $C'$ correct in $C$}\\\\
\forall \st{o}{C}.\\\\
bases(C) = C_1 \ldots C_n \\
P_0 = Q\\
\{ P[o/\texttt{this}, C/\theta] \}\ c[o/\texttt{this}]\ \{ P_n \ast \overline{o \deref f \mapsto \texttt{\_}} \ast \dyn{o}{C} \} \\\\
\{ P_n \ast \dyn{o_{C_n}}{C_n} \}\ o_{C_n} \deref \texttt{\textasciitilde}C_n()\ \{ P_{n-1} \} \\\\
\ldots \\\\
\{ P_1 \ast \dyn{o_{C_1}}{C_1} \}\ o_{C_1} \deref \texttt{\textasciitilde}C_1()\ \{ P_0 \}
}
{
\Gamma \vdash \texttt{\textasciitilde}C()\ \texttt{req}\ P\ \texttt{ens}\ Q\ \{ c \}\ \text{correct in $C$}
}
\end{mathpar}

\subsection{Program verification}
Verification of a class succeeds if verification for its constructor, destructor, and methods succeeds. We additionally require that a derived class overrides all methods declared in its base classes. This requirement renders our assumption sound that the dynamic type of the target object during verification of a destructor or method is the class type of the enclosing class it is declared in.

A program is correct if verification of all its classes succeeds, and its main command is verifiable given an empty heap.
\begin{mathpar}
\inferrule
{
prog = \overline{class}\ c\\
\vdash \overline{class}\ \text{correct}\\
\vdash \{ \texttt{true} \}\ c\ \{ \texttt{true} \}\\
}
{
\vdash \text{program correct}
}
\end{mathpar}

\begin{theorem}[Soundness]
Given that the program is correct, the main command, when executed in the empty heap, does not get stuck (i.e.~it either terminates or diverges):
$$\vdash \mathrm{program\ correct} \land prog = \overline{class}\ c \Rightarrow \emptyset, c \Downarrow \_ \lor \emptyset, c\,\texttt{div}$$
\end{theorem}

\section{Example proof outline}\label{sec:proof_outline}
This section shows an example in our formal language, annotated with its proof outline. It illustrates a \emph{node} class \texttt{N} which inherits from both a \emph{target} class \texttt{T} and \emph{source} class \texttt{S}. A target and source can have a source and target, respectively. A node is initially its own target and source. 

The example illustrates dynamic dispatch during construction and shows that our program logic is applicable in the presence of multiple inheritance. The main command shows how our proof system can handle polymorphic deletion of objects. The proof outline for \texttt{T} is symmetric to the one shown in \texttt{S}, and is therefore omitted. Empty bodies implicitly contain a \texttt{skip} command.

\begin{lstlisting}[style=proof]
class S {
  t := null;

  pred $Sok() = \exists t.\ this \deref t \mapsto t$; pred $sdyn(dt) = \dyn{this}{dt}$;

  S() req $\true$ ens $this \deref \Base{sdyn}{S}(\texttt{S}) \ast this \deref \Base{Sok}{S}()$ {
    $\{ \true \ast this \deref t \mapsto \nul \ast \Dyn{this}{S} \}$
    $\{ this \deref \Base{sdyn}{S}(\texttt{S}) \ast this \deref \Base{Sok}{S}() \}$
  }

  virtual ~S() req $this \deref \base{sdyn}{\theta}(\theta) \ast this \deref \base{Sok}{\theta}()$ ens $\true$ {
    $\{ \Dyn{this}{S} \ast  \exists t.\ this \deref t \mapsto t \}$
  }

  virtual setTarget(t) req $this \deref \base{Sok}{\theta}()$ ens $this \deref \base{Sok}{\theta}()$ {
    $\{ \exists lt.\ this \deref t \mapsto lt \}$
      $\{ this \deref t \mapsto lt \}$ this $\derefc$ t := t $\{ this \deref t \mapsto t \}$
    $\{ \exists lt.\ this \deref t \mapsto lt \}$
  }
};

class T {
  s := null;
  pred $Tok() = \exists s.\ this \deref s \mapsto s$; pred $tdyn(dt) = \dyn{this}{dt}$;
  T() {}
  virtual ~T() {}
  virtual setSource(s) { this $\derefc$ s := s }
};

class N : S, T {
  pred $Sok() = \Base{this}{S} \deref \Base{Sok}{S}() \ast \Base{this}{T} \deref \Base{Tok}{T}()$;
  pred $sdyn(dt) = \Base{this}{S} \deref \Base{sdyn}{S}(dt) \ast \Base{this}{T} \deref \Base{tdyn}{T}(dt)$;
  pred $Tok() = \Base{this}{S} \deref \Base{Sok}{S}() \ast \Base{this}{T} \deref \Base{Tok}{T}()$;
  pred $tdyn(dt) = \Base{this}{S} \deref \Base{sdyn}{S}(dt) \ast \Base{this}{T} \deref \Base{tdyn}{T}(dt)$;

  N() req $\true$ ens $this \deref \Base{sdyn}{N}(\texttt{N}) \ast this \deref \Base{Sok}{N}()$ :
    $\{ \true \}$
      $\{ \true \}$ S() $\{ \Base{this}{S} \deref \Base{sdyn}{S}(\texttt{S}) \ast \Base{this}{S} \deref \Base{Sok}{S}() \}$
    $\{ \Base{this}{S} \deref \Base{Sok}{S}() \ast \Dyn{\Base{this}{S}}{S} \}$
    $,$
    $\{ \Base{this}{S} \deref \Base{Sok}{S}() \}$
      $\{ \true \}$ T() $\{ \Base{this}{T} \deref \Base{tdyn}{T}(\texttt{T}) \ast \Base{this}{T} \deref \Base{Tok}{T}() \}$
    $\{ \Base{this}{S} \deref \Base{Sok}{S}() \ast \Base{this}{T} \deref \Base{Tok}{T}() \ast \Dyn{\Base{this}{T}}{T} \}$
  {
    $\{ \Base{this}{S} \deref \Base{Sok}{S}() \ast \Base{this}{T} \deref \Base{Tok}{T}() \ast \Dyn{this}{N} \}$
    $\{ \Dyn{this}{N} \ast this \deref \Base{Sok}{N}() \}$
      $\{ \exists C.\, \dyn{this}{C} \ast this \deref \base{Sok}{C}() \}$
      this $\derefc$ setTarget((T *) this);
      $\{ \dyn{this}{C} \ast this \deref \base{Sok}{C}() \}$
    $\{ \Dyn{this}{N} \ast this \deref \Base{Sok}{N}() \}$
    $\{ \Dyn{this}{N} \ast \Base{this}{S} \deref \Base{Sok}{S}() \ast \Base{this}{T} \deref \Base{Tok}{T}() \}$
    $\{ \Dyn{this}{N} \ast this \deref \Base{Tok}{N}() \}$
      $\{ \exists C.\, \dyn{this}{C} \ast this \deref \base{Tok}{C}() \}$
      this $\derefc$ setSource((S *) this)
      $\{ \dyn{this}{C} \ast this \deref \base{Tok}{C}() \}$
    $\{ \Dyn{this}{N} \ast this \deref \Base{Tok}{N}() \}$
    $\{ \Dyn{this}{N} \ast \Base{this}{S} \deref \Base{Sok}{S}() \ast \Base{this}{T} \deref \Base{Tok}{T}() \}$
    $\{ \Dyn{this}{N} \ast this \deref \Base{Sok}{N}() \}$
    $\{ \Dyn{\Base{this}{S}}{N} \ast \Dyn{\Base{this}{T}}{N} \ast this \deref \Base{Sok}{N}() \}$
    $\{ \Base{this}{S} \deref \Base{sdyn}{S}(\texttt{S}) \ast \Base{this}{T} \deref \Base{tdyn}{T}(\texttt{T}) \ast this \deref \Base{Sok}{N}() \}$
    $\{ this \deref \Base{sdyn}{N}(\texttt{N}) \ast this \deref \Base{Sok}{N}() \}$
  }

  virtual ~N() req $this \deref \base{sdyn}{\theta}(\theta) \ast this \deref \base{Sok}{\theta}()$ ens $\texttt{true}$ {
    $\{ \Base{this}{S} \deref \Base{sdyn}{S}(\texttt{N}) \ast \Base{this}{T} \deref \Base{tdyn}{T}(\texttt{N}) \ast this \deref \Base{Sok}{N}() \}$
    $\{ \Dyn{\Base{this}{S}}{N} \ast \Dyn{\Base{this}{T}}{N} \ast this \deref \Base{Sok}{N}() \}$
    $\{ \Base{this}{S} \deref \Base{Sok}{S}() \ast \Base{this}{T} \deref \Base{Tok}{T}() \ast \Dyn{this}{N} \}$
  }
    $\{ \Base{this}{S} \deref \Base{Sok}{S}() \ast \Base{this}{T} \deref \Base{Tok}{T}() \ast \Dyn{\Base{this}{T}}{T} \}$
    $\{ \Base{this}{S} \deref \Base{Sok}{S}() \ast \Base{this}{T} \deref \Base{Tok}{T}() \ast \Base{this}{T} \deref \Base{tdyn}{T}(\texttt{T}) \}$
      $\{ \Base{this}{T} \deref \Base{tdyn}{T}(\texttt{T}) \ast \Base{this}{T} \deref \Base{Tok}{T}() \}$ $\color{gray}\Base{this}{T}\deref\texttt{\textasciitilde T}()$ $\{ \true \}$
    $\{ \Base{this}{S} \deref \Base{Sok}{S}() \}$
    $,$
    $\{ \Base{this}{S} \deref \Base{Sok}{S}() \ast \Dyn{\Base{this}{S}}{S} \}$
    $\{ \Base{this}{S} \deref \Base{sdyn}{S}(\texttt{S}) \ast \Base{this}{S} \deref \Base{Sok}{S}() \}$
      $\{ \Base{this}{S} \deref \Base{sdyn}{S}(\texttt{S}) \ast \Base{this}{S} \deref \Base{Sok}{S}() \}$ $\color{gray}\Base{this}{S}\deref\texttt{\textasciitilde S}()$ $\{ \true \}$
    $\{ \true \}$

  virtual setTarget(t) req $this \deref \base{Sok}{\theta}()$ ens $this \deref \base{Sok}{\theta}()$ {
    $\{ \Base{this}{S} \deref \Base{Sok}{S}() \ast \Base{this}{T} \deref \Base{Tok}{T}() \}$
      $\{ \Base{this}{S} \deref \Base{Sok}{S}() \}$ this $\derefc$ S::setTarget(t) $\{ \Base{this}{S} \deref \Base{Sok}{S}() \}$
    $\{ \Base{this}{S} \deref \Base{Sok}{S}() \ast \Base{this}{T} \deref \Base{Tok}{T}() \}$
  }

  virtual setSource(s) req $this \deref \base{Tok}{\theta}()$ ens $this \deref \base{Tok}{\theta}()$ {
    $\{ \Base{this}{S} \deref \Base{Sok}{S}() \ast \Base{this}{T} \deref \Base{Tok}{T}() \}$
      $\{ \Base{this}{T} \deref \Base{Tok}{T}() \}$ this $\derefc$ T::setSource(s) $\{ \Base{this}{T} \deref \Base{Tok}{T}() \}$
    $\{ \Base{this}{S} \deref \Base{Sok}{S}() \ast \Base{this}{T} \deref \Base{Tok}{T}() \}$
  }
};

$\{ \true \}$
  $\{ \true \}$
  let n := new N() in
  $\{ n \deref \Base{sdyn}{N}(\texttt{N}) \ast n \deref \Base{Sok}{N}() \ast \Cted{n}{N} \}$
  let s := (S*) n in
  $\{ s \deref \Base{sdyn}{N}(\texttt{N}) \ast s \deref \Base{Sok}{N}() \ast \Cted{s}{N} \}$
    $\{ \exists C.\ \cted{s}{C} \ast s \deref \base{sdyn}{C}(C) \ast s \deref \Base{Sok}{C}() \}$
    delete s
    $\{ \true \}$
  $\{ \true \}$
$\{ \true \}$
\end{lstlisting}
The proof that the specification of \texttt{\textasciitilde N} implies the specification of \texttt{\textasciitilde T}, can be constructed as follows:
\begin{flalign*}
\inferrule*[Right=AMovePred]
{
    \inferrule*[Right=APredDef]
    {
        \inferrule*[Right=APredDef]
        {\color{NavyBlue}
            \{ this \deref \Base{sdyn}{N}(\texttt{N}) \ast this \deref \Base{Sok}{N}() \}
            \texttt{\color{black}\_} 
            \{ \texttt{true} \}
        }
        {\color{NavyBlue}
            \left\{
            {
                \begin{aligned}
                \Base{this}{s} \deref \Base{sdyn}{S}(\texttt{N}) \ast \Base{this}{T} \deref \Base{tdyn}{T}(\texttt{N}) \\
                \ast\,\Base{this}{S} \deref \Base{Sok}{S}() \ast \Base{this}{T} \deref \Base{Tok}{T}()
                \end{aligned}
            }
            \right\}
            \texttt{\color{black}\_} 
            \{ \texttt{true} \}
        }
    }
    {\color{NavyBlue}
        \{ this \deref \Base{tdyn}{N}(\texttt{N}) \ast this \deref \Base{Tok}{N}() \}
        \texttt{\color{black}\_} 
        \{ \texttt{true} \}
    }
}   
{\color{NavyBlue}
    \{ \Base{this}{T} \deref \Base{tdyn}{N}(\texttt{N}) \ast \Base{this}{T} \deref \Base{Tok}{N}() \}
    \texttt{\color{black}\_} 
    \{ \texttt{true} \}
}
\end{flalign*}
The behavioral suptyping proofs for the specifications of \texttt{setSource} and \texttt{setTarget}, and the proof that the specification of \texttt{\textasciitilde N} implies the specification of \texttt{\textasciitilde S}, can be established trivially using assertion weakening rule \textsc{AMovePred}.

\section{Related Work}\label{sec:related-work}
Parkinson and Bierman's work \cite{ParkinsonMatthew2005Slaa,ParkinsonMatthewJ.2008Slaa} introduces abstract predicate families. Their proof system allows a derived class to extend a base class, restrict the behavior of its base class, and alter the behavior of the base class while preserving behavioral subtyping. Method specifications consist of a dynamic and static specification, used for dynamically and statically dispatched calls, respectively. We derive these specifications from the same specification, using special variable $\theta$. Their proof system only accounts for single inheritance without the presence of virtual destructors.

\citet{RamananandroTahina2012AMSf} define operational semantics for a subset of \Cpp{}, including construction and destruction in the presence of multiple inheritance and virtual methods that are dynamically dispatched. Their semantics encode the evolution of an object's dynamic type during construction and destruction. However, they only consider \emph{stack-allocated} objects. This means that the concrete dynamic type of an object is always statically known at the point of its destruction.

\citet{van2009separation} extend the work of Parkinson and Bierman to a separation logic for object-oriented programs with multiple inheritance and virtual methods calls that are dynamically dispatched. They only consider virtual inheritance, which means that an object cannot have two base subobjects of the same class type. Furthermore, their logic does not support destructors, so polymorphic deletion is not considered. In their proof system, the dynamic type of an object is fixed after allocation, whereas we model the evolution of the dynamic type of an object during its construction and destruction.

BRiCk \cite{BRiCk}, built upon the separation logic of Iris \cite{jung2018iris}, is a program logic for \Cpp{}. The Frama-Clang plugin of Frama-C \cite{KirchnerFlorent2015FAsa} enables analysis of \Cpp{} programs, supporting the ACSL specification language. Both tools support dynamic dispatching and model the evolution of an object's dynamic type through its construction and destruction. However, at the time of writing, no literature on these tools’ approaches has appeared.

\section{Conclusion}\label{sec:conclusion}
In this paper we proposed a separation logic for modular verification of programs where virtual method calls are bound to different methods at different points during the construction and destruction of objects. Additionally, we support polymorphic destruction where the static type of an object is a supertype of its dynamic type.

We defined the operational semantics of our language related to allocation and deallocation, construction and destruction, and method dispatching, and listed the corresponding proof rules for verification.

Next, we illustrated an example program annotated with a proof outline, to support our verification approach. This example indicates that our separation logic can be used to verify \Cpp{} dynamic binding in the presence of multiple inheritance. To our knowledge, we are the first to define a Hoare logic which reflects \Cpp{}'s semantics of dynamic binding in the presence of constructors an destructors.

We implemented our approach \cite{verifast-tests} as part of our effort to extend our VeriFast tool for semi-automated modular formal verification of C and Java programs with support for \Cpp{}. The implementation in VeriFast additionally supports bases that are non-polymorphic. One limitation is that our current operational semantics and separation logic does not consider virtual inheritance.

\bibliographystyle{ACM-Reference-Format}
\bibliography{bibl}

\appendix
\section{Operational semantics}
The operational semantics of expressions, commands, and constructor and destructor invocations are defined by mutual induction:
\begin{mathparpagebreakable}
\inferrule[OLookup]
{
h, e \Downarrow  h' \uplus \lbbar o \deref f \mapsto v \rbbar, o
}
{
h, e \deref f \Downarrow  h' \uplus \lbbar o \deref f \mapsto v \rbbar, v
}
\and
\inferrule[ODeleteNull]
{
h, e \Downarrow  h', \texttt{null}
}
{
h, \texttt{delete}(e) \Downarrow  h'
}
\and
\inferrule[OVal]
{}
{
h, v \Downarrow h, v
}
\and
\inferrule[OUpdate]
{
h, e \Downarrow  h', o \\
 h', e' \Downarrow h'' \uplus \lbbar o \deref f \mapsto v \rbbar, v'
}
{
h, e \deref f \coloneqq e' \Downarrow h'' \uplus \lbbar o \deref f \mapsto v' \rbbar
}
\and
\inferrule[OLet]
{
h, e \Downarrow  h', v \\
 h', c[v/x] \Downarrow h''
}
{
h, \texttt{let}\ x \coloneqq e\ \texttt{in}\ c \Downarrow h''
}
\and
\inferrule[OSeq]
{
h, c \Downarrow h'\\
h', c' \Downarrow h''
}
{
h, c; c' \Downarrow h''
}
\and
\inferrule[OSkip]
{}
{
h, \texttt{skip} \Downarrow h
}
\and
\inferrule[OUpcast]
{
h, e \Downarrow  h', o \\
\st{o}{C} \\
C' \in bases(C) \\
}
{
h, ({C'*})\ e \Downarrow  h', o_{C'}
}
\and
\inferrule[OStaticDispatch]
{
\texttt{class}\ C \cdots {\{} \cdots \texttt{virtual}\ m(\overline{x}) {\{} c {\}} \cdots {\}} \\
h, e \Downarrow  h', o \\
\st{o}{C} \\
 h', \overline{e} \Downarrow h'', \overline{v} \\
h'', c[o/\texttt{this},\overline{v}/\overline{x}] \Downarrow h'''
}
{
h, e \deref C\text{::}m(\overline{e}) \Downarrow h'''
}
\and
\inferrule[ODynamicDispatch]
{
\texttt{class}\ C \cdots {\{} \cdots \texttt{virtual}\ m(\overline{x}) {\{} c {\}} \cdots {\}} \\
h, e \Downarrow  h', o \\
 h', \overline{e} \Downarrow h'', \overline{v} \\
dtype(o, C) \subseteq h'' \\
o' = o \downarrow C \\
h'', c[o'/\texttt{this},\overline{v}/\overline{x}] \Downarrow h'''
}
{
h, e \deref m(\overline{e}) \Downarrow h'''
}
\and
\inferrule[OConstruct]
{
\texttt{class}\ C : C_1 \ldots C_n {\{} \overline{f \coloneqq \texttt{null};} \cdots C(\overline{x}) : C_1(\overline{e_1}) \ldots C_n(\overline{e_n}) {\{} c {\}} \cdots {\}} \\
h, \overline{e} \Downarrow h_0, \overline{v} \\\\
h_0, o_{C_1} \deref C_1(\overline{e_1}[o/\texttt{this},\overline{v}/\overline{x}]) \Downarrow h_1 \uplus dtype(o_{C_1}, C_1) \\\\
\vdots \\\\
h_{n-1}, o_{C_n} \deref C_n(\overline{e_n}[o/\texttt{this},\overline{v}/\overline{x}]) \Downarrow h_n \uplus dtype(o_{C_n}, C_n) \\\\
h_n \uplus \lbbar \overline{o \deref f \mapsto \texttt{null}} \rbbar \uplus dtype(o, C), c[o/\texttt{this},\overline{v}/\overline{x}] \Downarrow  h'
}
{
h, o \deref C(\overline{e}) \Downarrow  h'
}
\and
\inferrule[ONew]
{
o = (id: C{*}) \\
\texttt{alloc}(id) \notin h \\
h \uplus \lbbar \texttt{alloc}(id) \rbbar, o \deref C(\overline{e}) \Downarrow  h'
}
{
h, \texttt{new}\ C(\overline{e}) \Downarrow  h' \uplus \lbbar \texttt{cted}(o, C) \rbbar, o
}
\and
\inferrule[ODestruct]
{
\texttt{class}\ C : C_1 \ldots C_n {\{} \overline{f \coloneqq \texttt{null};} \cdots \texttt{virtual}\ \texttt{\textasciitilde}C() {\{} c {\}} \cdots {\}} \\
h, c[o/\texttt{this}] \Downarrow h_n \uplus dtype(o, C) \uplus \lbbar \overline{o \deref f \mapsto v} \rbbar \\\\
h_n \uplus dtype(o_{C_n}, C_n), o_{C_n} \deref \texttt{\textasciitilde}C_n() \Downarrow h_{n-1} \\\\
\vdots \\\\
h_1 \uplus dtype(o_{C_1}, C_1), o_{C_1} \deref \texttt{\textasciitilde}C_1() \Downarrow h_0
}
{
h, o \deref \texttt{\textasciitilde}C() \Downarrow h_0
}
\and
\inferrule[ODelete]
{
o' = o \downarrow C \\
h, e \Downarrow  h'  \uplus \lbbar \texttt{cted}(o', C) \rbbar, o \\
h', o' \deref \texttt{\textasciitilde}C() \Downarrow h''
}
{
h, \texttt{delete}(e) \Downarrow h''
}
\end{mathparpagebreakable}

\section{Assertion semantics}
The semantics of assertions are defined as follows:
\begin{equation*}
\begin{array}{l l l}
I,h \vDash b & \Leftrightarrow & b = \texttt{true} \\
I,h \vDash P \ast Q & \Leftrightarrow & 
\begin{array}{@{}l@{}}
\exists h_1,h_2. h = h_1 \uplus h_2 {} \wedge {} \\
\quad I,h_1 \vDash P \wedge I,h_2 \vDash Q 
\end{array}
\\
I,h \vDash P \wedge Q & \Leftrightarrow & I,h \vDash P \wedge I,h \vDash Q \\
I,h \vDash P \vee Q & \Leftrightarrow & I,h \vDash P \vee I,h \vDash Q \\
I,h \vDash \exists x.\ P & \Leftrightarrow & \exists \nu.\ I,h \vDash P[\nu/x] \\
I,h \vDash o \deref p(C,\overline{\nu}) & \Leftrightarrow & \exists o'.\;o \downarrow C = o' \land (h, o', p, C, \overline{\nu}) \in I \\
I,h \vDash \texttt{cted}(o, C) & \Leftrightarrow & \exists o'.\;o \downarrow C = o' \land \cted{o'}{C} \in h \\
I,h \vDash \dyn{o}{C} & \Leftrightarrow & dtype(o, C) \subseteq h \\
I,h \vDash o \deref f \mapsto v & \Leftrightarrow & o \deref f \mapsto v \in h
\end{array}
\end{equation*}
where $I,h \vDash P$ means that assertion $P$ is satisfied, given heap $h$ and interpretation of predicates $I$. Cases not listed are false.

\section{Proof rules}
We define evaluation contexts for expressions and commands as follows:
\begin{align*}
K_e & \Coloneqq \bullet \mid K_e \deref f \mid \texttt{new}\ C(\overline{v}\ K_e\ \overline{e}) \mid (C{*}) K_e
\\
K_c & \Coloneqq \bullet \mid \texttt{delete}\ K_e \mid K_e \deref f \coloneqq e \mid o \deref f \coloneqq K_e \mid K_e \deref C{::}m(\overline{e}) \\
&\mid o \deref C{::}m(\overline{v}\ K_e\ \overline{e}) \mid K_e \deref m(\overline{e}) \mid o \deref m(\overline{v}\ K_e\ \overline{e})
\end{align*}
We use the notation $K[e]$ to denote the context $K$ with expression $e$ substituted for the hole $\bullet$.

\begin{mathparpagebreakable}
\inferrule[HFrame]
{
\{P\}\ c\ \{Q\}
}
{
\{P \ast R\}\ c\ \{Q \ast R\}
}
\and
\inferrule[HConseq]
{
P \Rightarrow_a P' \\
\{ P' \}\ c\ \{ Q' \} \\
Q' \Rightarrow_a Q
}
{
\{ P \}\ c\ \{ Q \}
}
\and
\inferrule[HNull]
{}
{\{ \true{} \}}\ \nul{}\ \{ \texttt{result} = \nul{} \}
\and
\inferrule[HPointer]
{}
{\{ \true{} \}\ o\ \{ \texttt{result} = o \}}
\and
\inferrule[HLookup]
{}
{
\{ o \deref f \mapsto v \}\ o \deref f\ \{ o \deref f \mapsto v \wedge \texttt{result} = v \}
}
\and
\inferrule[HUpdate]
{}
{
\{ o \deref f \mapsto \texttt{\_} \}\ o \deref f \coloneqq v\ \{ o \deref f \mapsto v \}
}
\and
\inferrule[HLet]
{
\{ P \}\ e\ \{ Q \}\\
\forall v.\ \{ Q[v/\texttt{result}] \}\ c[v/x]\ \{ R \}
}
{
\{ P \}\ \texttt{let}\ x \coloneqq e\ \texttt{in}\ c\ \{ R \}
}
\and
\inferrule[HSeq]
{
\{ P \}\ c\ \{ Q \}\\
\{ Q \}\ c'\ \{ R \}
}
{
\{ P \}\ c; c'\ \{ R \}
}
\and
\inferrule[HSkip]
{}
{
\{ P \}\ \texttt{skip}\ \{ P \}
}
\and
\inferrule[HContext]
{
\{ P \}\ e\ \{ Q \}\\
\forall v.\ \{ Q[v/\texttt{result}] \}\ K[v]\ \{ R \}
}
{
\{ P \}\ K[e]\ \{ R \}
}
\and
\inferrule[HConsContext]
{
\{ P \}\ e\ \{ Q \}\\
\forall v.\ \{ Q[v/\texttt{result}] \}\ o \deref C(\overline{v}\ v\ \overline{e})\ \{ R \}
}
{
\{ P \}\ o \deref C(\overline{v}\ e\ \overline{e})\ \{ R \}
}
\and
\inferrule[HConstruct]
{
\texttt{class}\ C \cdots {\{} \cdots C(\overline{x})\ \texttt{req}\ P\ \texttt{ens}\ Q \cdots {\}} \in \Gamma
}
{
\{ P[\overline{v}/\overline{x}]) \}\ o\deref C(\overline{v})\ \{ Q[\overline{v}/\overline{x}, o/\texttt{this}] \}
}
\and
\inferrule[HNew]
{
\texttt{class}\ C \cdots {\{} \cdots C(\overline{x})\ \texttt{req}\ P\ \texttt{ens}\ Q \cdots {\}} \in \Gamma
}
{
\{ P[\overline{v}/\overline{x}] \}\ \texttt{new}\ C(\overline{v})\ \{ Q[\overline{v}/\overline{x}, \texttt{result}/\texttt{this}] \ast \texttt{cted}(\texttt{result}, C) \}
}
\and
\inferrule[HDestruct]
{
\texttt{class}\ C \cdots {\{} \cdots \texttt{virtual \textasciitilde} C()\ \texttt{req}\ P\ \texttt{ens}\ Q \cdots {\}} \in \Gamma
}
{
\{ P[o/\texttt{this}, C/\theta] \}\ o \deref \texttt{\textasciitilde}C()\ \{ Q \}
}
\and
\inferrule[HDeleteNull]
{}
{
\{\true\}\ \texttt{delete}(\texttt{null})\ \{\true\}
}
\and
\inferrule[HDelete]
{
\texttt{class}\ C \cdots {\{} \cdots \texttt{virtual \textasciitilde} C()\ \texttt{req}\ P\ \texttt{ens}\ Q \cdots {\}} \in \Gamma \\
\st{o}{C}
}
{
\{ \texttt{cted}(o, C') \ast P[o/\texttt{this}, C'/\theta] \}\ \texttt{delete}(o)\ \{ Q \}
}
\and
\inferrule[HStaticDispatch]
{
\texttt{class}\ C \cdots {\{} \cdots \texttt{virtual}\ m(\overline{x})\ \texttt{req}\ P\ \texttt{ens}\ Q \cdots {\}} \in \Gamma \\
\st{o}{C}
}
{
\{ P[o/\texttt{this},C/\theta,\overline{v}/\overline{x}]\}\ o \deref C\text{::}m(\overline{v})\ \{ Q[o/\texttt{this}, C/\theta, \overline{v}/\overline{x}] \}
}
\and
\inferrule[HDynamicDispatch]
{
\texttt{class}\ C \cdots {\{} \cdots \texttt{virtual}\ m(\overline{x})\ \texttt{req}\ P\ \texttt{ens}\ Q \cdots {\}} \in \Gamma \\
\st{o}{C}
}
{
\{ \dyn{o}{C'} \wedge P[o/\texttt{this}, C'/\theta, \overline{v}/\overline{x}] \}\ o \deref m(\overline{v})\ \{ Q[o/\texttt{this}, C'/\theta, \overline{v}/\overline{x}] \}
}
\and
\inferrule[HExists]
{
\forall v.\ \{ P[v/x] \}\ c\ \{ Q \}
}
{
\{ \exists x.\ P \}\ c\ \{ Q \}
}
\and
\inferrule[HUpcast]
{
\st{o}{C}\\
C \in bases(C')
}
{
\{ P[o_C/\texttt{result}] \}\ (C{*})\ o\ \{ P \}
}
\end{mathparpagebreakable}

\subsection{Destructor override check}
\begin{mathpar}
\inferrule
{
\texttt{class}\ C \cdots {\{} \cdots \texttt{virtual}\ \texttt{\textasciitilde}C()\ \texttt{req}\ P\ \texttt{ens}\ Q \cdots {\}} \in \Gamma\\
\texttt{class}\ C' \cdots {\{} \cdots \texttt{virtual}\ \texttt{\textasciitilde}C'()\ \texttt{req}\ P'\ \texttt{ens}\ Q' \cdots {\}} \in \Gamma\\
\Gamma \vdash \{ P \}\texttt{\_}\{ Q \} \refines{C}{C'} \{ P' \}\texttt{\_}\{ Q' \}\\
{
\begin{array}{r}
\forall C'' \in bases(C').\,\Gamma \vdash \text{override of destructor in $C''$ correct in C}
\end{array}
}
}
{
\Gamma \vdash \text{override of destructor in $C'$ correct in $C$}
}
\end{mathpar}

\subsection{Class verification}
\begin{mathpar}
\inferrule
{class = \texttt{class}\ C \cdots {\{} \cdots ctor\ dtor\ \overline{meth}\ {\}}\\
\Gamma \vdash ctor\ \text{correct in $C$}\\
\Gamma \vdash dtor\ \text{correct in $C$}\\
\Gamma \vdash \overline{meth}\ \text{correct in $C$}
}
{\Gamma \vdash class\ \text{correct}}
\end{mathpar}

\section{Soundness}


Due to the fact that our assertion language does not allow predicate assertions in negative positions (i.e.~under negation or on the left-hand side of implication), we have the following property:
\begin{lemma}\label{lem:asn-mono}
The semantics of assertions is monotonic in the predicate interpretation $I$:
$$I \subseteq I' \land I, h \vDash P \Rightarrow I', h \vDash P$$
\end{lemma}
\begin{proof}
By induction on the structure of $P$.
\end{proof}

We define a function $F$ on predicate interpretations as follows:
\begin{equation*}
F(I) = \left\{
(h, o, p, C, \overline{v})\ \middle|\ 
\begin{array}{@{}l@{}}
\texttt{class}\ C\ \cdots\ \{\ \cdots\ \texttt{pred}\ p(\overline{x}) = P;\ \cdots\ \}\\
\land\, I, h \vDash P[o/\texttt{this}, \overline{v}/\overline{x}]
\end{array}
\right\}
\end{equation*}

We define the program's predicate interpretation $I_\text{program}$ by
$I_\text{program} = \bigcap\{I\;|\;F(I) \subseteq I\}$.
By the Knaster-Tarski theorem, $I_\text{program}$ is a fixpoint of $F$: $F(I_\text{program}) = I_\text{program}$.\footnote{It is in fact the least fixpoint.} We use notation $h \vDash P$ to mean $I_\text{program}, h \vDash P$.

\begin{lemma}[Soundness of assertion weakening]
$$P \Rightarrow_a Q \land h \vDash P \Rightarrow h \vDash Q$$
\end{lemma}
\begin{proof}
By induction on the derivation of $P \Rightarrow_a Q$.
\end{proof}

We define semantic counterparts of the correctness judgments of our proof system as follows:
\begin{equation*}
\begin{array}{l}
\vDash \{P\}\ e\ \{Q\} \Leftrightarrow\\
\quad \left(\forall h, h_\text{f}.\;h\vDash P \Rightarrow \begin{array}{l}
h \uplus h_\text{f}, e\;\texttt{div} \lor {}\\
\begin{array}{@{}r@{}}
\exists h', v.\;h \uplus h_\text{f}, e \Downarrow h' \uplus h_\text{f}, v \\
\land\, h' \vDash Q[v/\texttt{result}]
\end{array}
\end{array}\right)\\
\\
\vDash \{P\}\ c\ \{Q\} \Leftrightarrow\\
\quad \left(\forall h, h_\text{f}.\;h\vDash P \Rightarrow \begin{array}{l}
h \uplus h_\text{f}, c\;\texttt{div} \lor {}\\
\exists h'.\;h \uplus h_\text{f}, c \Downarrow h' \uplus h_\text{f} \land h' \vDash Q
\end{array}\right)\\
\\
\vDash \{P\}\ o\deref C(\overline{e})\ \{Q\}\Leftrightarrow\\
\quad \left(\forall h, h_\text{f}.\;h\vDash P \Rightarrow \begin{array}{l}
h \uplus h_\text{f}, o\deref C(\overline{e})\;\texttt{div} \lor {}\\
\exists h'.\;h \uplus h_\text{f}, o\deref C(\overline{e}) \Downarrow h' \uplus h_\text{f} \land h' \vDash Q
\end{array}\right)\\
\\
\vDash \{P\}\ o\deref \texttt{\textasciitilde}C()\ \{Q\}\Leftrightarrow\\
\quad \left(\forall h, h_\text{f}.\;h\vDash P \Rightarrow \begin{array}{l}
h \uplus h_\text{f}, o\deref \texttt{\textasciitilde}C()\;\texttt{div} \lor {}\\
\exists h'.\;h \uplus h_\text{f}, o\deref \texttt{\textasciitilde}C() \Downarrow h' \uplus h_\text{f} \land h' \vDash Q
\end{array}\right)
\end{array}
\end{equation*}

\begin{lemma}{Soundness of \textsc{HContext}}\label{lemma:hcontext-soundness}
If $\vDash \{P\}\ e\ \{Q\}$ and $\forall v.\;\vDash \{Q[v/\texttt{result}]\}\ K[v]\ \{R\}$ then
$\vDash \{P\}\ K[e]\ \{R\}$.
\end{lemma}
\begin{proof}
By induction on the structure of $K$.
\end{proof}

\begin{assumption}
The program is correct:
$$\vdash \text{program correct}$$
\end{assumption}

\begin{lemma}[Main Soundness Lemma]
\begin{equation*}
\begin{array}{l}
\forall h, h_\text{f}, P, Q.\; h \vDash P \Rightarrow\\
\quad (\forall e.\;\{P\}\ e\ \{Q\} {} \land {} \\
\quad \quad (\nexists h', v.\;h\uplus h_\text{f}, e \Downarrow h' \uplus h_\text{f}, v \land h' \vDash Q[v/\texttt{result}]) \Rightarrow \\
\quad \quad \quad h \uplus h_\text{f}, e\;\texttt{div}) {} \land {}\\
\quad (\forall c.\;\{P\}\ c\ \{Q\} \land (\nexists h'.\;h\uplus h_\text{f}, c \Downarrow h' \uplus h_\text{f} \land h' \vDash Q) \Rightarrow\\
\quad \quad h \uplus h_\text{f}, c\;\texttt{div}) {} \land {}\\
\quad (\forall o, C, \overline{e}.\;\{P\}\ o\deref C(\overline{e})\ \{Q\} {} \land {} \\
\quad \quad (\nexists h'.\;h\uplus h_\text{f}, o\deref C(\overline{e}) \Downarrow h' \uplus h_\text{f} \land h' \vDash Q) \Rightarrow \\
\quad \quad \quad h \uplus h_\text{f}, o\deref C(\overline{e})\;\texttt{div}) {} \land {}\\
\quad (\forall o, C.\;\{P\}\ o\deref \texttt{\textasciitilde}C()\ \{Q\} {} \land {} \\ 
\quad \quad (\nexists h'.\;h\uplus h_\text{f}, o\deref C(\overline{e}) \Downarrow h' \uplus h_\text{f} \land h' \vDash Q) \Rightarrow \\
\quad \quad \quad h \uplus h_\text{f}, o\deref \texttt{\textasciitilde}C()\;\texttt{div})
\end{array}
\end{equation*}
\end{lemma}
\begin{proof}
By mutual co-induction and, nested inside of it, induction on the derivation of the correctness judgment. We elaborate a few cases:
\begin{itemize}
\item Case \textsc{HDynamicDispatch}. Assume the following:
$$\begin{array}{c}
c = o \deref m(\overline{v})\\
\st{o}{C}\\
\texttt{class}\ C \cdots {\{} \cdots \texttt{virtual}\ m(\overline{x})\ \texttt{req}\ P_C\ \texttt{ens}\ Q_C \cdots {\}}\\
P = \dyn{o}{D} \wedge P_C[o/\texttt{this}, D/\theta, \overline{v}/\overline{x}]\\
Q = Q_C[o/\texttt{this}, D/\theta, \overline{v}/\overline{x}]\\
\texttt{class}\ D \cdots {\{} \cdots \texttt{virtual}\ m(\overline{x})\ \texttt{req}\ P_D\ \texttt{ens}\ Q_D\ \{c_m\} \cdots {\}}
\end{array}$$
By $h \vDash P$, we have $dtype(o, D) \subseteq h$ and $h \vDash P_C[o/\texttt{this}, D/\theta, \overline{v}/\overline{x}]$.
Let $o' = o \downarrow D$. By the correctness of method $m$ in class $D$, we have
$$
\begin{array}{c}
\{P_D[o'/\texttt{this}, D/\theta, \overline{v}/\overline{x}]\} \\
c_m[o'/\texttt{this}, \overline{v}/\overline{x}] \\
\{Q_D[o'/\texttt{this}, D/\theta, \overline{v}/\overline{x}]\}
\end{array}
$$
By the fact that $m$ in $D$ correctly overrides $m$ in $C$, we have
$$\{P_D\}\_\{Q_D\} \refines{D}{C} \{P_C\}\_\{Q_C\}$$ It follows that
$$
\begin{array}{c}
\{P_C[o/\texttt{this}, D/\theta, \overline{v}/\overline{x}]\} \\
c_m[o'/\texttt{this}, \overline{v}/\overline{x}] \\
\{Q_C[o/\texttt{this}, D/\theta, \overline{v}/\overline{x}]\}
\end{array}
$$
The relevant inference rule for divergence of dynamically dispatched method calls is as follows:
$$\inferrule[ODynamicDispatchDiv3]
{
\texttt{class}\ C \cdots {\{} \cdots \texttt{virtual}\ m(\overline{x}) {\{} c {\}} \cdots {\}} \\
h, e \Downarrow  h', o \\
o' = o \downarrow C \\\\
 h', \overline{e} \Downarrow h'', \overline{v} \\
dtype(o, C) \subseteq h'' \\
h'', c[o'/\texttt{this},\overline{v}/\overline{x}]\,\texttt{div}
}
{
h, e \deref m(\overline{e})\,\texttt{div}
}$$
We apply this rule to the goal, which reduces the goal to $h, c_m[o'/\texttt{this},\overline{v}/\overline{x}]\,\texttt{div}$. We now apply the coinduction hypothesis. We are now left with the job of proving that the body does not terminate, assuming that the call does not terminate. Instead, we prove that the call terminates, assuming that the body terminates. We conclude that proof by applying \textsc{ODynamicDispatch}.

\item Case \textsc{HConsContext}. Assume a constructor argument list $\overline{v}\,e\,\overline{e}$. By the induction hypothesis corresponding to the first premise of \textsc{HConsContext}, we have that evaluation of $e$ either terminates or diverges.
\begin{itemize}
\item Assume $e$ terminates with a value $v$. By the induction hypothesis corresponding to the second premise of \textsc{HConsContext}, we have that $o\deref C(\overline{v}\,v\,\overline{e})$ either terminates or diverges.
\begin{itemize}
\item Assume $o\deref C(\overline{v}\,v\,\overline{e})$ terminates. This must be by an application of \textsc{OConstruct}. Therefore, it must be that $\overline{e}$ all terminate. It follows that $o \deref C(\overline{v}\,e\,\overline{e})$ terminates.
\item Assume $o\deref C(\overline{v}\,v\,\overline{e})$ diverges. Given that $e$ terminates, we can easily prove that $o\deref C(\overline{v}\,e\,\overline{e})$ diverges.
\end{itemize}
\item Assume $e$ diverges. Then $o\deref C(\overline{v}, e, \overline{e})$ diverges.
\end{itemize}
\item Case \textsc{HContext}. We apply Lemma~\ref{lemma:hcontext-soundness} and use the induction hypotheses to discharge the resulting subgoals.\footnote{To see that this preserves productivity of the coinductive proof, notice that Lemma~\ref{lemma:hcontext-soundness} is \emph{size-preserving}: given approximations up to depth $d$ of the proof trees for the lemma's premises, the lemma produces a proof tree of depth at least $d$.}
\end{itemize}
\end{proof}

\end{document}